%% file: main.tex
\documentclass[conference,compsoc]{IEEEtran}

\usepackage{algorithm}
\usepackage{algorithmic}
\usepackage{multirow}
\usepackage{pifont} 
\usepackage{amsmath}
\usepackage{amsthm}
\usepackage{amssymb}
\usepackage{xcolor}
\usepackage{afterpage}
\newtheorem{theorem}{Theorem}

\AtBeginDocument{%
  }

\ifCLASSINFOpdf
  \usepackage[pdftex]{graphicx}
\else
\fi
\usepackage{xspace}
\usepackage{subcaption}
\usepackage{hyperref}
\usepackage{tablefootnote}
\usepackage{booktabs}
\newcommand{\name}{RepoMark\xspace}

\newcommand{\revise}[1]{{\color{black}#1}}

\begin{document}

\title{RepoMark: A Data-Usage Auditing Framework for Code Large Language Models}

\author{
\IEEEauthorblockN{
Wenjie Qu$^1$,
Yuguang Zhou$^1$,
Bo Wang$^1$,
Yuexin Li$^1$,
Lionel Z. Wang$^2$,
Jinyuan Jia$^3$, Jiaheng Zhang$^1$}
\IEEEauthorblockA{$^1$ National University of Singapore, $^2$ Nanyang Technological University, $^3$ Penn State University}
}

\maketitle
\input{abs}



%



\input{intro}

\input{related_work}

\input{threat}

\input{method}

\input{experiment}

\input{discuss}

\input{conclusion}

\clearpage
\bibliographystyle{IEEEtran}
\bibliography{ref}

\input{appendix}

\end{document}

%% file: abs.tex
\begin{abstract}
The rapid development of Large Language Models (LLMs) for code generation has transformed software development by automating coding tasks with unprecedented efficiency. 
 However, the training of these models on open-source code repositories (e.g., from GitHub) raises critical ethical and legal concerns, particularly regarding data authorization and open-source license compliance. Developers are increasingly questioning whether model trainers have obtained proper authorization to use repositories for training, especially given the lack of transparency in data collection.

To address these concerns, we propose a new data marking framework \name to audit the data usage of code LLMs. Our method enables auditors to verify whether their code has been used in training, while ensuring semantic preservation, imperceptibility, and a theoretical guarantee on false detection rate (FDR). By generating multiple semantically equivalent code variants, \name introduces data marks into the code files, and during detection, \name leverages a new ranking-based hypothesis test to detect model behavior difference on trained data. Compared to prior data auditing approaches, \name significantly enhances data efficiency, allowing effective auditing even when the user's repository possesses only a small number of code files.

Experiments demonstrate that RepoMark achieves a detection success rate over 90\% on small code repositories under a strict FDR guarantee of 5\%. This represents a significant advancement over existing data marking techniques, all of which only achieve accuracy below 55\% under identical settings. These results validate RepoMark as a robust, theoretically sound, and promising solution for enhancing transparency in code LLM training, which can safeguard the rights of code authors.

\end{abstract}

%% file: intro.tex
\section{Introduction}
\label{sec:intro}
The rapid advancement of Large Language Models (LLMs) in code generation has revolutionized software development, enabling developers to automate coding tasks with unprecedented efficiency. Open-source communities such as GitHub, which maintain a large set of code repositories, have become the main source of  training data for code LLMs~\cite{code_llm_2,copilot,code_llm_3}.  However, the use of publicly available code repositories to train these models has raised critical ethical and legal questions.

A primary concern is data authorization—specifically, whether model developers have obtained proper consent from code authors to use their repositories for training purposes. While source code itself may not always constitute personal data, developers' rights regarding the usage and distribution of their code align with broader data protection principles established by regulations such as the General Data Protection Regulation (GDPR)~\cite{GDPR} in Europe, which grant data owners the right to know how their data is used. Nevertheless, the training of machine learning models often suffers from a pervasive lack of transparency. Model trainers rarely disclose the details of the origin of their training data~\cite{generalccs24}, and the opaque data collection processes employed in training code LLMs make it difficult to audit whether proper authorization was obtained from code authors for model training. This lack of transparency is especially concerning given the emergence of commercial products such as GitHub Copilot~\cite{copilot_commerce}, Amazon CodeWhisperer~\cite{code_llm_2}, Cursor~\cite{code_llm_3}, Tabnine~\cite{code_llm_4}, and Codeium~\cite{code_llm_5}. 
Exploiting source code without authorization to train models for commercial purposes directly violates the rights and contributions of original authors and undermines the spirit of the open-source community.

Given these ethical and legal concerns, tracing the data-usage of code LLMs has become a pressing challenge. The most prevalent approach to address this challenge is data-usage auditing~\cite{generalccs24}, which enables auditors to verify whether their data has been used to train a machine learning (ML) model. Existing data-usage auditing methods can be categorized into two classes: membership inference and data marking. Membership inference~\cite{mia1,mia3,mia4} infers if a data sample is a member of an ML model’s training set without modifying the training data.  
    In contrast, data marking~\cite{radio,backdoor_wm_1,generalccs24,datamark_llm1} modifies the data samples prior to publication, enabling detection methods to exploit statistical signals introduced during marking to detect data-usage in training.  
In general, data marking leverages more information and   often achieves higher detection accuracy~\cite{generalccs24} than membership inference. Therefore, we design an auditing method for code LLMs based on the data marking paradigm in this paper.

 To effectively audit code LLMs, our auditing method should satisfy four properties: (1) \emph{Preservation of Semantics}: The modifications introduced by the marking method must retain the original code's semantics. (2) \emph{Imperceptibility}: The code marks should be difficult for model trainers to identify, as they are incentivized to remove the marked code files to avoid being caught. (3) \emph{Data Efficiency}: In the real world, many individual repositories only have 10--50 code files. The method should maintain sufficient accuracy in detecting model training on such small-scale repositories. (4) \emph{FDR Guarantee}: 
    The method should provide a guarantee on the false detection rate (FDR). FDR is defined as the probability of incorrectly flagging a target model as having been trained on a repository when no such training occurred. A theoretical FDR guarantee is imperative, as false accusations against the model trainer could lead to significant reputational harm and ethical breaches.

However, no existing data marking scheme satisfies all four properties simultaneously.
First, many methods—e.g., CodeMark~\cite{sun2023codemark} and Huang et al.~\cite{generalccs24} (CCS’24)—are not sufficiently data efficient and deliver suboptimal accuracy when auditing small individual repositories, as shown in our experiments.  
Second, as noted by \cite{generalccs24}, most auditing approaches do not provide a provable guarantee on FDR.

\revise{To simultaneously achieve the four desired properties, we first introduce a new code auditing framework that provides enhanced data efficiency and a provable FDR guarantee, thereby addressing the latter two properties. Building upon this framework, we design a concrete marking algorithm specifically tailored to the code domain, which together satisfy all desired properties.

}

\noindent \textbf{Code auditing framework.}  \revise{In our code auditing framework, we consider three parties: the code author, the model trainer, and the auditor (e.g., a court). The code author owns a repository composed of multiple source files. Before releasing it publicly, the author embeds marks into the code to enable future auditing.
To insert these marks, the author generates $m$ semantically equivalent variants for each code file and randomly selects one version to publish while keeping the remaining variants private. Later, the model trainer collects large amounts of online code without authorization and uses it to train a code LLM, which may include the marked repository. 
During the detection phase, the auditor computes the code LLM-based loss of the published version and all its variants, ranks the published version within each set, and performs a hypothesis test over the aggregated ranks across all files to determine whether the model was trained on the marked repository. }

\revise{The core idea of our framework is to capture and aggregate the statistical behavior differences between the trained and untrained cases for each code file.} In particular, due to the randomness in the selection of the published version, if the code LLM is not trained on a code file, the rank of the published version will be uniformly distributed among $\{1,\cdots,m\}$. In contrast, if the code LLM is trained on a code file, then the rank is likely to be biased towards $1$. Our detection algorithm leverages the summation of the ranks across different code files to amplify this bias for detection. Due to the uniform property of the rank distribution under the untrained case, the FDR of our method is theoretically upper-bounded.

\noindent \textbf{Code marking algorithm.}
The remaining challenge is how to design a dedicated marking algorithm and its corresponding detection algorithm for code that fits into our framework. The critical property our marking algorithm has to satisfy is that the $m$ variants of code it generates all preserve the program semantics, while the difference between the $m$ variants and the original code should be imperceptible to the model trainer.

To preserve the program semantics of marked code, our marking algorithm focuses on renaming local variables. For each code file, we select one variable to rename. A similar variable set of size $m$ is generated based on token likelihoods computed by an oracle code LLM. The selected variable is then renamed to a variable randomly chosen from this set. The different $m$ variable names chosen from this set correspond to the $m$ versions of the code file required by our code auditing framework. 
The integration of code variable renaming into our data marking framework yields \name, which satisfies all four desired properties.

We further improve the data efficiency of \name by injecting multiple marks into long code files.  When the model is not trained on the marked file, the randomness in publication selection ensures that, during detection, the ranks of different injected marks within the same file are independent of each other and uniformly distributed over $\{1, \dots, m\}$. This property guarantees the correctness of the detection of multiple marks per file and greatly improves the data  efficiency of our framework.

\name consistently achieves high detection accuracy across different code LLMs and datasets. Under a 5\% FDR guarantee, it correctly identifies over 90\% of the repositories used for training, whereas the best baseline detects only about 54\%. Moreover, \name maintains strong imperceptibility—the marked code has an average perplexity of 1.11, closely matching the unmarked code’s 1.04.

Our contributions are as follows:\\
\begin{itemize}
    \item We propose a general code auditing framework with a theoretical FDR guarantee and better data efficiency.
    
    \item We propose a marking algorithm for code that can generate $m$ variants for a single code file while preserving program semantics and imperceptible to the model trainer. Incorporating the marking algorithm into our general framework, we introduce \name, a new data-usage auditing framework for code LLMs that simultaneously achieves the four desired properties.

    \item We validate the effectiveness of \name through extensive experiments. Experiments show that it achieves high accuracy in detecting the training data-usage of code LLMs and significantly outperforms all previous baselines.

\end{itemize}

%% file: related_work.tex
\section{Background and related work}
\subsection{Code large language models}

 Modern large language models (LLMs) typically utilize the Transformer architecture~\cite{transformer}.  The LLM's input context consists of a sequence of tokens, denoted as $\left\langle s_1, s_2, \cdots, s_L\right\rangle$. These LLMs predict subsequent tokens in an autoregressive manner. We denote the vocabulary of an LLM as $\mathcal{V}$. The LLM predicts the next token $s_{L+1}$ by first mapping the token sequence to a logit vector $z\in \mathbb{R}^{|\mathcal{V}|}$:\\
$$z=\mathsf{LLM}(\left\langle s_1, s_2, \cdots, s_L\right\rangle)$$

Then the model uses a decoding strategy to decide the next token $s_{L+1}$ based on the logit vector. Commonly adopted strategies include top-$k$ sampling~\cite{topk} and nucleus sampling~\cite{topp}. 
The newly generated token $s_{L+1}$ is then appended to the model input context to generate the next token. This autoregressive process repeats until a stop condition is met, typically when a special end-of-sequence token is generated or when a maximum sequence length is reached.

Code LLMs are large language models designed to assist developers in software engineering tasks such as code completion, debugging, and documentation. 
They share the same architecture as general-purpose LLMs but are trained on a large number of code repositories.

\subsection{Membership inference attack}
Membership inference attack (MIA) is a type of confidentiality attack in machine learning, which aims to infer whether a particular data sample has been used to train a target ML model~\cite{mia1,mia3,mia4,mia5}. Existing MIA methods can be classified into loss-based attacks~\cite{mia3,mia4,mia6,mia7} and shadow model-based attacks~\cite{mia1,mia5}.

MIA can serve as a passive data-usage auditing method, as it does not require any prior modification or marking of the training data. 
Most existing MIAs operate at the instance level, aiming
to identify whether a single example was part of the training data. 
 A few works~\cite{maini2021dataset,mainillm,tongmuch} have generalized MIA to the \emph{dataset} level, detecting whether a given dataset was used to train a model by aggregating instance-level MIA results across data points.  This line of work is more closely related to our scenario, where an auditor seeks to detect whether a code LLM has been trained on a given code repository.  
However, applying MIA to audit the data-usage of code LLMs faces a critical limitation: although thresholds for MIA methods can be calibrated under experimental conditions by setting an empirical FDR using known member and non-member labels, in real-world auditing scenarios, where such labels are unavailable, there is no guarantee on the FDR of MIA methods.

\subsection{Data marking}

Data marking is a type of proactive technique that allows a data owner to audit the use of their data in a target ML model~\cite{generalccs24,backdoor_wm_1,backdoor_wm_2,backdoor_wm_3,radio,datamark_llm1}. Such methods usually consist of a marking algorithm
that embeds marks into data and a detection algorithm that tests whether the data has been used to train a given model. Prevalent data marking methods can be divided into two categories: backdoor marking and contrastive marking. 

\revise{\subsubsection{Backdoor marking}

Backdoor marking has been extensively studied in prior work~\cite{backdoor_wm_1,backdoor_wm_2,backdoor_wm_3}. 
In this paradigm, the data owner poisons a small fraction of the training set by injecting a trigger pattern and relabeling these samples to a chosen target class—mirroring standard backdoor training~\cite{badnets,chen2017targeted,lin2020composite,cheng2021deep}. Then, the data owner releases the marked dataset. 
During detection, the owner adds the trigger to new data points and examines whether the model outputs the chosen target label.

To the best of our knowledge, CodeMark~\cite{sun2023codemark} is the only work that designs a backdoor marking scheme specifically for code LLMs.  It embeds triggers as co-occurrence patterns within the source code.
Detection is then performed by testing on new code files whether these patterns frequently appear in the model’s outputs.

Applying backdoor marking to code faces an inherent data efficiency challenge: in practical auditing, the target repository may contain very few files,  yielding an extremely low poisoning ratio and,  consequently, weak detection power. 
Empirically, as shown in Table~\ref{tab:main}, CodeMark~\cite{sun2023codemark} detects only a small fraction of marked repositories.   
Fundamentally, backdoor marking methods rely on learning a \emph{strong trigger signal} that can generalize across samples, which is hard to satisfy under very low poisoning rates.

\subsubsection{Contrastive marking} \label{sec:datamark_llm}

To improve data efficiency, another line of research explores \emph{contrastive marking}~\cite{datamark_llm1,generalccs24,sablayrolles2020radioactive,wang2024diagnosis}. In this paradigm, the data owner generates multiple semantically equivalent variants for each data point and releases only one of them publicly. 
During detection, the owner compares the target model’s behavior (e.g., loss) on the published version against its unpublished counterparts. 
If the published version was not used during training, the losses across all variants should be similar; otherwise, the published version typically exhibits a smaller loss.

Compared to backdoor methods that rely on learning strong trigger signals, contrastive marking exploits subtle statistical differences in the model's behavior between seen and unseen data.  Consequently,  when the data owner controls very few data samples (e.g., a small code repository), contrastive marking achieves substantially higher detection accuracy, as demonstrated in our experiments (Table~\ref{tab:main}).   

Next, we briefly introduce two contrastive marking methods~\cite{datamark_llm1,generalccs24} that support auditing LLMs. Wei et al.~\cite{datamark_llm1} create multiple variants for each data file by appending randomly generated character sequences.   
During detection, the owner feeds each original (unmodified) data file into the target model and measures the number of matched tokens for each corresponding public and private version. 
They then empirically construct the distribution of matched token lengths under the null hypothesis—i.e., assuming the model was not trained on the published data—using the matched token lengths of the unpublished sequences.
The key intuition is that if the model has not seen the published version during training, its output probabilities for different random sequences should be nearly uniform; conversely, if it has been trained on the published sequence, it will exhibit a higher likelihood of reproducing that sequence.
A hypothesis test is subsequently performed to determine whether the matched token length on the published sequences significantly exceeds the null-hypothesis distribution, thereby indicating that the model has likely been trained on them. 
However, this method is difficult to apply to code LLMs, since appending random character sequences (e.g., in code comments) produces conspicuous patterns that can be easily detected and automatically removed by model trainers. 

More recently, Huang et al.~\cite{generalccs24} proposed a general contrastive marking framework that supports both image classifiers and language models.  
To the best of our knowledge, it is the only data-usage auditing framework that provides a rigorous FDR guarantee.   
Before data publication, the data owner generates two slightly modified versions of each raw data file, publishes one version chosen uniformly at random, and keeps the other private. 
During detection, the owner computes the loss for both versions and counts the number of published files whose loss is smaller than that of their private counterparts. 
If the published data were not used for training, this count follows a Binomial distribution with mean $\frac{1}{2}$.  
A hypothesis test on this statistic then provides a provable FDR bound for detecting data-usage. 
However, this method is primarily designed for dataset-level auditing. When applied to individual repositories with only a few code files, it results in suboptimal detection accuracy, as demonstrated in our experiments (Table~\ref{tab:main}).

}

%% file: threat.tex
\section{Problem formulation}
\subsection{Threat model}

In our code auditing framework, we consider three parties: the code author, the model trainer, and the auditor (e.g., a third-party authority). The code author owns a code repository consisting of code files $\left\{x_1, x_2, \ldots, x_N\right\}$ that will be released online on a platform like GitHub. The model trainer aims to train a code model with strong coding capabilities on a large set of code data. The model trainer assembles a training dataset $\mathcal{D}$  by collecting code online from many repositories without authorization. As such, after publication, the code author's data constitutes a subset of the model trainer's collected dataset (i.e., the repository $\{x_1, x_2, \ldots, x_N\} \subseteq \mathcal{D}$). The code LLM trainer trains a model on $\mathcal{D}$ using a learning algorithm. Then the model trainer deploys the trained code LLM to provide services to consumers (e.g., to monetize the trained code LLM).

The auditor is responsible for determining whether the deployed model has been trained on the code author’s marked repository. During the detection phase, the auditor computes the loss based on the code LLM for the published version and all its private variants, ranks the published version within each set, and performs a hypothesis test over the aggregated ranks across all files to determine whether the model has been trained on the marked repository.

Following the threat model adopted in prior data-usage auditing works~\cite{generalccs24,datamark_llm1,radio}, we assume that the auditor has access to the code LLM's output logits during inference, but has no access to its internal architecture or parameters. The auditor has full access to the marked repository, including all code files, the corresponding renaming positions, and the generated variants. 

In closed-source code LLM auditing scenarios (e.g., commercial models such as Cursor~\cite{code_llm_3}), the auditor can be a third-party authority, such as a court. Although modern commercial models typically do not expose logits to end users, it is reasonable in copyright litigation (e.g., The New York Times v.s. OpenAI~\cite{nyt_oai}) for courts to request model providers to disclose such runtime information for auditing purposes. In contrast, in open-source code LLM auditing scenarios, the auditor can simply be the code author themselves, as the logits can be directly computed from the publicly available model.

\begin{figure*}[t]

\centering
\subfloat{\includegraphics[width=0.9\textwidth]{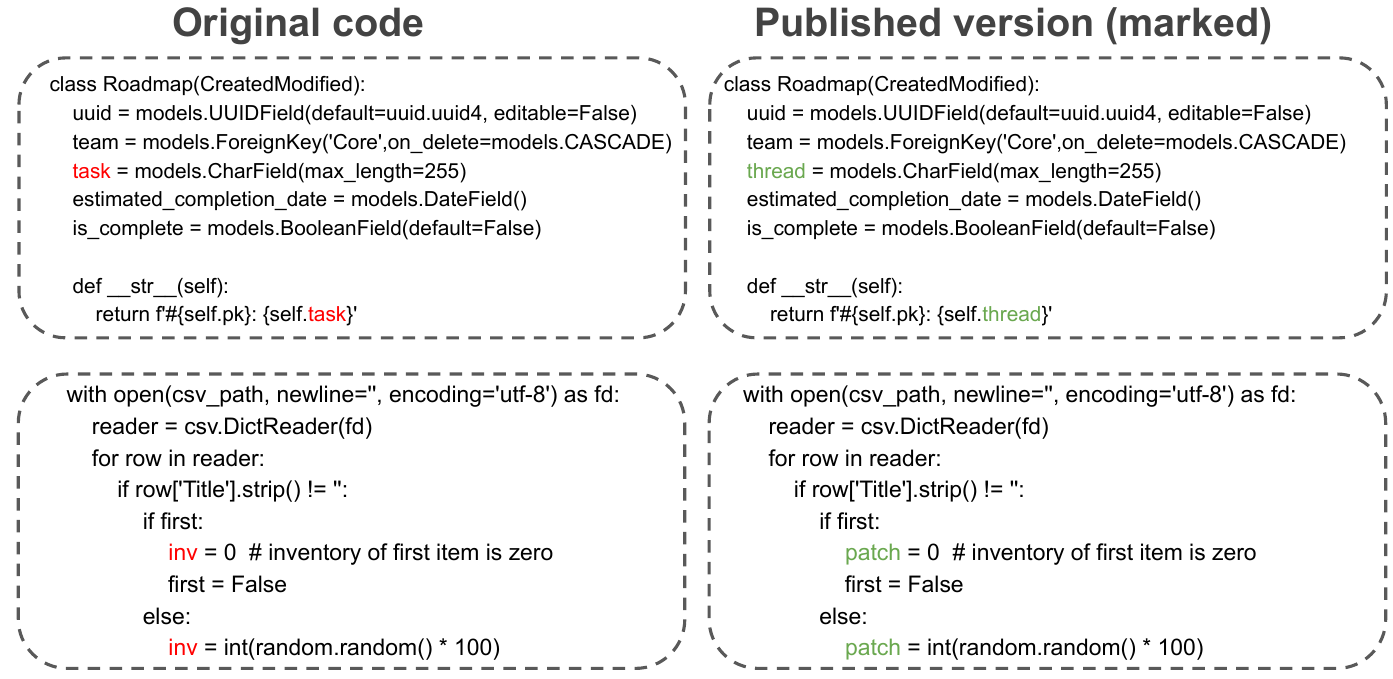}}
\caption{ The examples of code marked with \name. }
\label{fig:example}
\end{figure*}

\subsection{Design goals}
\label{sec:goal}
As mentioned in Section~\ref{sec:intro}, an effective code LLM data-usage auditing method should achieve the following properties:\\
\begin{itemize}
   \item Preservation of semantics: All modifications should preserve the original program’s semantics.
   
   \item Imperceptibility: The modifications to the code should be difficult for the trainer of the code LLM to detect.
    
\item Data Efficiency: The method should accurately detect training data-usage even when the repository contains only a small number of code files (e.g., 10–50). 

\item FDR Guarantee: The method should provide a statistical upper bound on its FDR, which is the probability of incorrectly identifying the code LLM as having been trained on the target repository when it was not.
    
\end{itemize}

An FDR guarantee of $p$ limits the probability of wrongly accusing a model of training on a protected repository to below $p$, which is critical for preventing false claims and associated ethical risks. Under a given FDR constraint, the effectiveness of auditing is measured by the detection success rate (DSR), defined as $\frac{\text{\# of correctly identified trained repositories}}{\text{\# of repositories actually trained on}}$.  A high DSR indicates that the method can detect most training repositories, which is essential for effective auditing.

Although probabilistic evidence is not as definitive as deterministic proof, it is often sufficient and widely accepted in legal practice. A classic example is DNA matching: it does not claim that ``this DNA definitely belongs to the suspect,'' but rather that ``the suspect’s DNA matches, and the probability of a random individual also matching is 1 in 10,000.'' Moreover, in copyright infringement cases, evidence only needs to meet the lower ``preponderance of the evidence'' standard~\cite{legal_dict} under US civil law, requiring just $>50\%$ likelihood of validity. All these examples demonstrate the practical potential of data-usage auditing methods that offer strong FDR guarantees when providing deterministic results is not feasible.


%% file: method.tex
\section{Methodology}

\subsection{New auditing framework}

 \revise{We first introduce a new, data-efficient contrastive marking framework with a provable FDR guarantee, designed to detect whether a code LLM has been trained on the marked code repository $\{x_1, \ldots, x_N\}$. This framework serves as the foundation of \name. }   The core idea is to generate multiple semantically equivalent versions for each code file, randomly select one version to publish, and use the sum of the loss ranks of the published versions to perform hypothesis testing for detection.

 At the marking phase, the code author creates $m$ different equivalent versions of each code file $x_i$, $i\in \{1,\cdots,N\}$, denoted as $x_i^1,x_i^2,\cdots,x_i^m$. For each $i\in \{1,\cdots,N\}$, the code author uniformly randomly chooses an index $v_i \stackrel{\$}{\leftarrow}\{1,2,\cdots,m\}$, and publishes $x_i^{v_i}$ while keeping all of the other versions of $x_i$ private. 
During the detection phase, for each code file $x_i$, we query the target code LLM with all $m$ of its versions and compute the corresponding losses, forming the set $\{\ell(x_i^1), \ldots, \ell(x_i^m)\}$. Our detection algorithm merely relies on the rank of $\ell(x_i^{v_i})$ in the loss of $m$ variants, namely, $\mathsf{rk}(\{\ell(x_i^1),\cdots,\ell(x_i^m)\}$,$\ell(x_i^{v_i}))$. Using ranks enables rigorous probabilistic guarantee analysis, which is difficult to achieve by analyzing loss values directly, as the latter requires assumptions about the distribution of the loss values across the $m$ variants. 

Our key theoretical insight is as follows: if the model is not trained on $x_i^{v_i}$, then $\mathsf{rk}(\{\ell(x_i^1),\cdots,\ell(x_i^m)\},\ell(x_i^{v_i}))$—the rank of $\ell(x_i^{v_i})$, is uniformly distributed among $\{1,2,\cdots,m\}$.  The uniform randomness stems from the random sampling of $v_i$, which removes the potential bias in different versions of $x_i$. However, if the model is trained on $x_i^{v_i}$, since the model is not trained on the other versions, $\ell(x_i^{v_i})$ is more likely to have a relatively small rank (e.g., less than $\frac{m}{2}$). 

Now, we can leverage hypothesis testing to build a detection algorithm with a rigorous FDR guarantee. The statistical quantity $S$ we consider for hypothesis testing is the sum of the ranks of all published data, formally defined as $S=\sum_{i=1}^N \mathsf{rk}(\{\ell(x^{1}_i),\cdots,\ell(x^{m}_i)\},\ell(x^{v_i}_i))$.

We set the rank sum threshold as $T$. 
If the model is not trained on the repository $\{x_1,\cdots,x_N\}$, then each $\mathsf{rk}(\{\ell(x^{1}_i),\cdots,\ell(x^{m}_i)\}$ $,\ell(x^{v_i}_i))$ is uniformly distributed among $\{1,\cdots,m\}$ due to the randomness of $v_i$. By the central limit theorem, $S$ will be close to its expectation $\frac{N(m+1)}{2}$ with high probability. Thus, for thresholds $T$ less than $\frac{N(m+1)}{2}$ with sufficient margin, if the model is not trained on the repository $\{x_1, \cdots, x_N\}$, the event $S > T$ occurs with high probability. In contrast, if the model is trained on the repository $\{x_1,\cdots,x_N\}$, for each sample $x_i$, $\mathsf{rk}(\{\ell(x^{1}_i),\cdots,\ell(x^{m}_i)\},\ell(x^{v_i}_i))$ is likely to be smaller than $\frac{m}{2}$, and $S$ would be much more likely smaller than $T$.  As such, the detection problem of our scheme can be formulated to test the following hypothesis:
 \begin{itemize}
     \item $H_0: \mathsf{rk}(\{\ell(x^{1}_i),\cdots,\ell(x^{m}_i)\},\ell(x^{v_i}_i))\sim \mathsf{Uniform}(\{1,2,\cdots,m\})$ \\
     \item $H_1: \mathsf{rk}\left(\left\{\ell(x_i^1), \cdots, \ell(x_i^m)\right\}, \ell(x_i^{v_i})\right)$ is biased towards  smaller values
 \end{itemize}

Under $H_0$, $S$ follows the distribution of the sum of $N$ i.i.d variables, each uniformly distributed among 
$\{1,2,\cdots,m\}$. As such, during detection, the auditor can reject $H_0$ or accept $H_0$ according to the value of $S$. The FDR of this testing procedure is provably upper-bounded, as the bound corresponds to the probability of rejecting $H_0$
  when it is actually true. 
  Given the upper-bounded FDR, rejecting  $H_0$
  implies that, with high probability, the target model has been trained on the published samples. The FDR guarantee could be adjusted by altering $T$.

\begin{figure*}[h]

\centering
\subfloat{\includegraphics[width=0.8\textwidth]{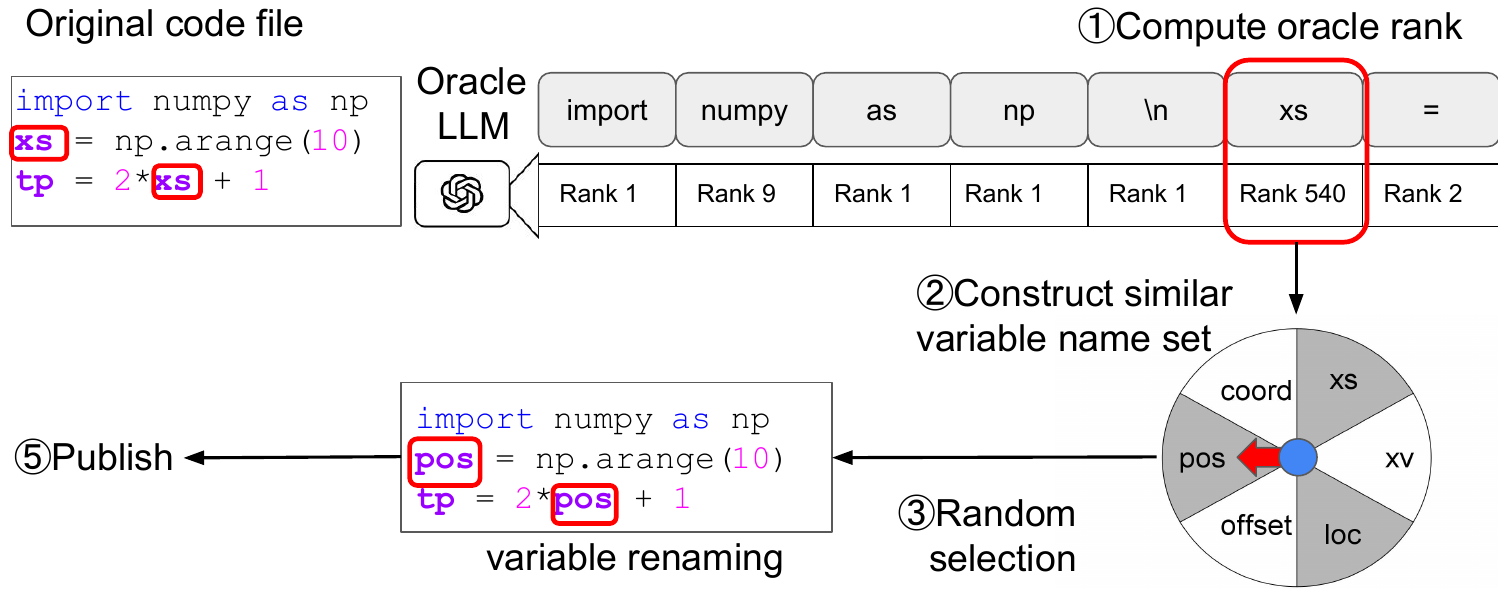}}
\caption{Illustration of \name's marking process for code.
This simplified figure highlights a single file with one variable modification for clarity. In practice, \name marks every file in the repository.}
\label{fig:marking}

\end{figure*}

\subsection{Semantic preservation and imperceptible code marking }

Under our code auditing framework, a core component is the design of an algorithm that generates semantically equivalent variants of each code file, while preserving both functionality and readability.
In this section, we focus on the design of this algorithm, which creates $m$ variants for each code file based on a single marking position.  We later extend this method to support multiple marking positions within a single file in Section~\ref{sec:scale_mark}, which further improves the data efficiency of our approach.

The preservation of code semantics is a critical property our marking algorithm has to achieve,  because disruptions to code semantics would negatively impact the repository's readability and functionality, hindering our scheme's real-world deployment.  Imperceptibility is also important for a practical marking algorithm, as the malicious model trainers are incentivized to remove the marks to avoid tracing.

To achieve both code semantics preservation and imperceptibility simultaneously, we propose a natural strategy of only renaming variables. In our marking algorithm, we only rename variables that are local variable and only consist of a single token. 
To apply a mark to a code file, we first select a variable to rename. Once a variable is selected, we use an oracle model (another code LLM) to propose alternative variable names, forming a similar variable name set of size $m$. Each of the $m$ versions of the code file is generated by renaming the selected variable with one of these $m$ alternatives.

A strawman approach is to randomly select some variables and construct the alternative list using their synonyms. However, this approach is limited by the typically small number of synonyms available for each variable.    Since higher values of $m$ lead to better detection accuracy, we instead focus on variables that admit a large set of alternative names.  Crucially, these alternatives must have similar predicted likelihoods under the oracle model. This requires the original variable name to have a relatively low predicted likelihood, so that more alternative names fall within a similar likelihood range.

This leads to a key insight: we should prefer variables whose names have a low likelihood when predicted by the oracle model.  However, a challenge arises because variables typically appear multiple times within a file, and their predicted likelihood varies across occurrences. To address this, we determine a variable's likelihood based on its first occurrence. This choice is motivated by an important observation: code LLMs find it significantly harder to predict a variable's name at its first appearance compared to later ones. This is intuitive—once a variable has been introduced, subsequent references are easier to guess, both for models and humans. Predicting the first appearance is much harder as it demands reasoning from local code patterns and broader coding conventions, rather than relying on repetition.

In conclusion, we focus on single-token variables whose first occurrence has a relatively high logits rank under the oracle model. Let $P$ denote the logits rank of the token corresponding to the variable name. We require that $P \geq R$, where $R$ is a predefined threshold, and we ensure $R \gg m$ to allow enough similar candidates. For each selected variable, we form its similar variable name set $T$ by collecting tokens whose logits ranks fall in the range $[P-\frac{m}{2}, P+\frac{m}{2})$.  Finally, to guarantee that renaming does not alter program semantics, we parse the code file using a lightweight static analysis tool to build its abstract syntax tree (AST). For each token in set $T$, we rename all occurrences of the target variable in the AST to produce each marked version of the code file.

Intuitively, our marking method can be viewed as a heuristic that renames variables whose original names ``surprise'' the oracle code LLM, replacing them with alternative names that have similar oracle rankings. The reason why we choose these ``surprising'' variables is that, if we choose variables that have small oracle ranks, to achieve imperceptibility, we can only choose other tokens that also have small rankings, resulting in a small search space of equivalent variants (i.e., small $m$). Choosing to rename these ``surprising'' variables allows for a much larger $m$, thereby increasing the information contained in the rank and improving detection accuracy.

 The imperceptibility property of this strategy can be explained from the perspective of the perplexity of the oracle model. The perplexity depends on the logits of each token. These $m$ different versions have similar logit ranks measured by the oracle model, and LLM logits typically follow a long-tail distribution. Since $P \gg m$, the logits of all $m$ different versions are likely to be very similar to each other. This indicates that, the change to the code caused by renaming variables is imperceptible measured by the perplexity of the oracle model.


\afterpage{\begin{figure*}[h]

\centering
\subfloat{\includegraphics[width=0.9\textwidth]{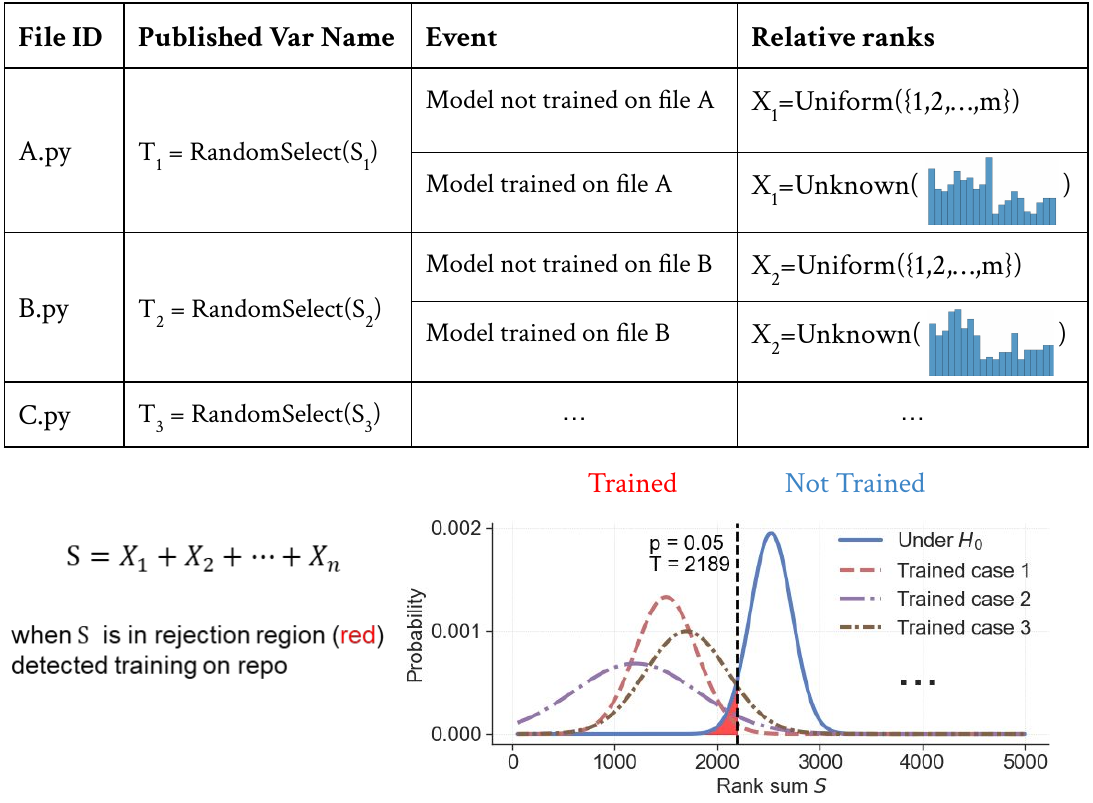}}
\caption{Illustration of RepoMark’s detection process for a single repository with multiple files (containing multiple marks). }
\label{fig:detection}
\end{figure*}}

\begin{figure}[h]
\centering
\subfloat{\includegraphics[width=0.35\textwidth]{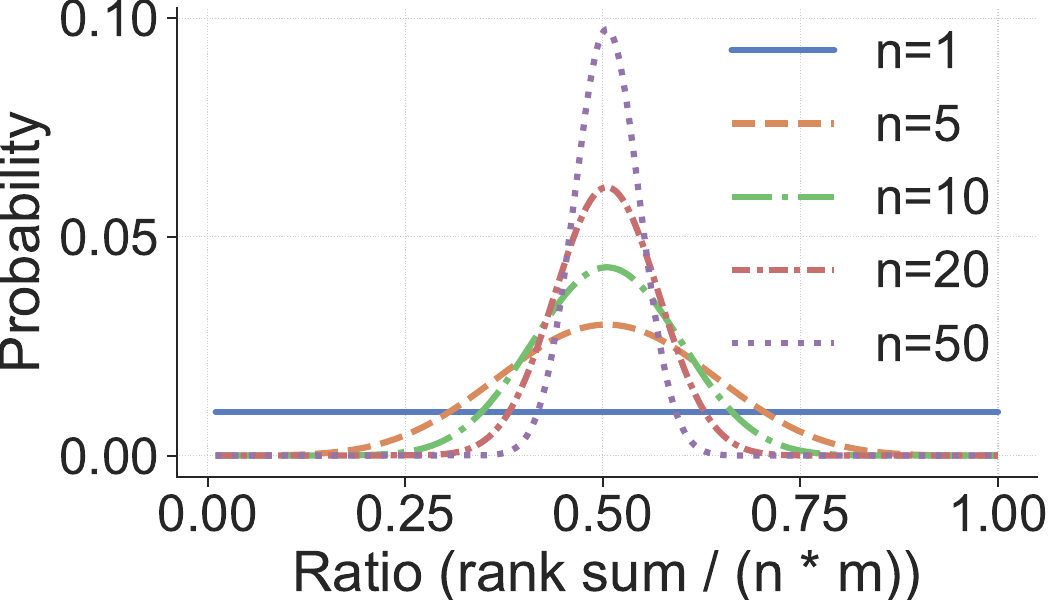}}
\caption{The distribution of the ratio of rank sum to $n*m$ under $H_0$ (the untrained case). With larger $n$, the distribution of the rank sum is more concentrated to $\frac{n(m+1)}{2}$, under $H_0$. }
\label{fig:distr}
\vspace{-1em}
\end{figure}

The complete marking procedure for a single code file is formalized in Algorithm~\ref{alg:replacement}. During repository-level auditing, this algorithm is applied iteratively to each file.

\subsection{Detection with  FDR guarantee}
\label{sec:fp_detection}

In this section, we provide a detailed description and analysis of \name's detection algorithm, with an illustration provided in Figure~\ref{fig:detection}. Its high-level idea is to capture and aggregate the statistical behavior differences between the trained and untrained cases for each code file.
 In particular, the auditor computes the loss ranks of the published version in different code files, and test whether the rank sum deviates significantly from the $H_0$ (untrained) case.

We first analyze the rank distribution of each code file under $H_0$ (untrained case).  We denote the code LLM weight as $w$, and it follows an unknown distribution $W$. Target model with weight $w$ accepts the token sequence $x$ as input and predicts the probability of variable name $t$ being the next token after $x$. For simplicity, we denote the corresponding cross-entropy loss as $f(w, x, t)$. We denote the $m$ versions of variable renamings as $\{t^1,\cdots,t^m\}$, and the published version's ID as $r\sim \mathsf{Uniform}(1,\cdots,m)$. For the rank of the published version $\mathsf{rk}(\{f(w,x,t^1),\cdots, f(w,x,t^m)\}, f(w,x,t^r))$, we have the following theorem:\\
\begin{theorem}
    Under $H_0$, $\mathsf{rk}(\{f(w,x,t^1),\cdots, f(w,x,t^m)\}$ $, f(w,x,t^r))$ is uniformly distributed among $\{1,\cdots,m\}$.
\end{theorem}

\begin{proof}
   Under $H_0$,  $w$ is independent with $r$.     $\forall k\in \{1,\cdots,m\}$, we have:\\
\begin{small}
    \begin{align}
    \begin{aligned}
             &\Pr(\mathsf{rk}(\{f(w,x,t^1),\cdots, f(w,x,t^m)\}, f(w,x,t^r))=k)\\
     =& \sum_{y=1}^m \sum_{w'\in W} \Pr(\mathsf{rk}(\{f(w',x,t^1),\cdots\}, f(w',x,t^y))=k) 
     \\
     &\Pr(r=y,w=w') \\
     =& \sum_{y=1}^m \sum_{w'\in W}  \Pr(\mathsf{rk}(\{f(w',x,t^1),\cdots\}, f(w',x,t^y))=k) \\
     &\cdot \frac{1}{m}  \Pr(w=w') \\
     =& \frac{1}{m} \sum_{w'\in W} \Pr(w=w') \\
     &\sum_{y=1}^m  \Pr(\mathsf{rk}(\{f(w',x,t^1),\cdots\}, f(w',x,t^y))=k)   
        \end{aligned}
        \label{eq:thm:core}
\end{align}  
\end{small}


Given $w',x$, $\{f(w',x,t^1),\cdots, f(w',x,t^m)\}$ is a deterministic set, with only one element in it having rank $k$.  Thus, $\forall k\in \{1,\cdots,m\}$, we have:\\
\begin{align*}
   & \sum_{y=1}^m  \Pr\left(\mathsf{rk}(\{f(w',x,t^1),\cdots\}, f(w',x,t^y))=k\right)=1\\
\end{align*}

Substituting the above equation into Equation~\ref{eq:thm:core}, we have:\\
\begin{align*}
 &\Pr(\mathsf{rk}(\{f(w,x,t^1),\cdots, f(w,x,t^m)\}, f(w,x,t^r))=k)\\
             =&\frac{1}{m} \sum_{w'\in W} \Pr(w=w') = \frac{1}{m}
\end{align*}
\end{proof}

The above theorem shows that when model $w$ is not trained on the published version $t^r$, the rank of the published version’s loss among the losses of all versions $\{t^1, \cdots, t^m\}$ is uniformly distributed over $\{1, \cdots, m\}$. This uniform randomness essentially stems from the uniform selection of $r$ from $\{1, \cdots, m\}$.

When $w$ is trained on the published version $t^r$, $w$ becomes dependent on $r$, and the rank is no longer uniformly distributed over $\{1,\cdots,m\}$. While the complex behavior of LLMs makes the exact rank distribution difficult to characterize, training should bias ranks toward smaller values.  In particular, the untrained mean rank is $\frac{m+1}{2}$, whereas the trained mean should be significantly smaller.

There are multiple code files in one repository. Therefore, we can leverage the rank sum to aggregate the message delivered by the mark in each code file. With more injected marks (larger $n$), the distribution of the rank sum of the $n$ marks under the untrained case is more concentrated around $\frac{n(m+1)}{2}$, as we can see in Figure~\ref{fig:distr}.

 We denote the rank sum of different mark positions as $S$ and the hyperparameter of rank sum threshold as $T$.  If the model is not trained on the target repository, by the central limit theorem, $S$ will be close to its expectation $\frac{n(m+1)}{2}$ with high probability. Therefore,  if the model is not trained on the target repository, for $T< \frac{n(m+1)}{2}$ with a sufficient gap, $\Pr(S\le T)$ will be very  small. In contrast, if the model is trained on the target repository, each $\mathsf{rk}({\ell(x^{1}_i),\cdots,\ell(x^{m}_i)},\ell(x^{v_i}_i))$ is much more likely to have a value smaller than $\frac{m}{2}$, and consequently $S$ is likely to be smaller than $T$. As such, the detection problem of our scheme can be formulated as the following hypothesis test:
 \begin{itemize}
     \item $H_0: \mathsf{rk}(\{\ell(x^{1}_i),\cdots,\ell(x^{m}_i)\},\ell(x^{v_i}_i))\sim \mathsf{Uniform}(\{1,2,\cdots,m\})$ \\
     \item $H_1: \mathsf{rk}(\{\ell(x_i^1), \cdots, \ell(x_i^m)\}, \ell(x_i^{v_i})) $ is biased towards smaller values 
 \end{itemize}

Under $H_0$, $S$ follows the distribution of the sum of $n$ i.i.d. random variables uniformly distributed among $\{1,2,\cdots,m\}$. The auditor can reject $H_0$ or accept $H_0$ based on whether $S> T$.  In other words, the auditor detects whether their repository was used to train the target code LLM  according to whether $S\le T$. The FDR of this test is provably upper-bounded by the Type I error—the probability that $H_0$ is incorrectly rejected—which corresponds to the probability that the sum of these 
$n$ uniform i.i.d. random variables is less than or equal to $T$.   Thus, the FDR guarantee could be controlled by altering the threshold $T$.

To determine the appropriate threshold, we can first compute the cumulative distribution function (CDF)  of $S$ under $H_0$, and then perform a binary search on the CDF table. The CDF of the above distribution can be computed using generating functions~\cite{gen_func,generatingfunctionology,genfunc}.
The generating function for a single uniform discrete variable over $\{1,2,\cdots,m\}$ is $G(x) = \frac{x + x^2 + \dots + x^m}{m}$. 
For the sum of $n$ i.i.d. discrete uniform random variables over $\{1,2,\cdots,m\}$, the generating function is $G_n(x) = \left(  \frac{x + x^2 + \dots + x^m}{m}  \right)^n$.

To extract the probability mass function (PMF), we expand the polynomial coefficients of \(G_n(x)\) into a sequence. We then extract the coefficient of \(x^t\), which gives \(\Pr(S = t)\): $\Pr(S = t) = [x^t] G_n(x)$. 
The polynomial coefficients of $G_n(x)$ can be efficiently computed using Fast Fourier transform~\cite{FFT}. Then we can compute the CDF table leveraging the computed PMF, and use binary search on the CDF array to find the largest $T$ such that $\Pr(S \le T) \leq p$.  The detection procedure of \name is presented in Algorithm~\ref{alg:detection}.

\subsection{Further improving the data efficiency of \name}
\label{sec:scale_mark}
A remaining problem is that, for a lengthy code file, our current design cannot effectively utilize its information redundancies, as it only injects one mark  for each file. To better utilize the information redundancy in lengthy code files, we further derive the case of renaming multiple different variables in a single file. We follow the notations used in Section~\ref{sec:fp_detection}. Under $H_0$, the core difference when injecting multiple marks is that, now $x$ is no longer a deterministic token sequence, but a random sequence, whose randomness is introduced by the random selection of previous renamed variables in the same file. Under $H_0$, $w$ is independent of $x$. Denote the possible set of context string $x$ as $\mathcal{X}$. We have:\\
\begin{small}
        \begin{align}
    \begin{aligned}
        &\Pr(\mathsf{rk}(\{f(w,x,t^1),\cdots, f(w,x,t^m)\}, f(w,x,t^r))=k)\\
     =& \frac{1}{m} \sum_{w'\in W} \Pr(w=w') \\ 
     &\sum_{y=1}^m  \Pr(\mathsf{rk}(\{f(w',x,t^1),\cdots\}, f(w',x,t^y))=k)\\
      =&\frac{1}{m} \sum_{w'\in W} \Pr(w=w') \sum_{p\in \mathcal{X}}  \Pr(x=p)   \\
      & \sum_{y=1}^m\Pr(\mathsf{rk}(\{f(w',p,t^1),\cdots\}, f(w',p,t^y))=k)
    \end{aligned}
    \label{eq:more_position_case}
\end{align}
\end{small}

Given $p$, $y$, $w'$, each $\{f(w',p,t^1),\cdots\}$ is a deterministic set.  Therefore,\\
$$\sum_{y=1}^m  \Pr\left(\mathsf{rk}(\{f(w',p,t^1),\cdots\}, f(w',p,t^y))=k \right)=1$$

Combining this with Equation~\ref{eq:more_position_case}, we have:\\
    \begin{align*}
    \begin{aligned}
        &\Pr(\mathsf{rk}(\{f(w,x,t^1),\cdots, f(w,x,t^m)\}, f(w,x,t^r))=k)\\
      =&\frac{1}{m} \sum_{w'\in W} \Pr(w=w') \sum_{p\in \mathcal{X}}  \Pr(x=p) \\=&\frac{1}{m}  
    \end{aligned}
    \label{eq:more_position_case_uniform}
\end{align*}

The above equation shows that if the model weights $w$ are not trained on the file, then when multiple marks are injected into the same code file, the ranks of the published versions for each injected mark are still i.i.d variables, uniformly distributed over $\{1, \cdots, m\}$. This property ensures that \name can support injecting multiple marks (i.e., renaming multiple variables) within a single file, making our method highly scalable to long code files.


We use mark sparsity threshold $K$ to control how many marks we should inject into one code file. Given mark sparsity threshold $K$, at most one mark can be injected per $K$ lines of code.
In other words, for a code file with $L$ lines of code, we only add at most $\lfloor\frac{L}{K} \rfloor$ marks. Following the previous construction, we only rename local variables whose first occurrence has an oracle model rank greater than or equal to the threshold $R$.  If there are more than $\lfloor\frac{L}{K} \rfloor$  candidate variables that satisfy this property, we randomly select $\lfloor\frac{L}{K} \rfloor$ variables among them. By choosing an appropriate $K$ value, we can effectively balance imperceptibility and detection accuracy. When a smaller $K$ is chosen, the detection would be more accurate as the injected mark number $n$ is larger. However, in the meantime, it would be easier for the model trainer to detect  the mark.




%% file: experiment.tex
\section{Experiments}

\subsection{Experimental setup}

 \begin{figure*}[t]
 \centering
\subfloat[CodeParrot Dataset]{\includegraphics[width=0.32\textwidth]{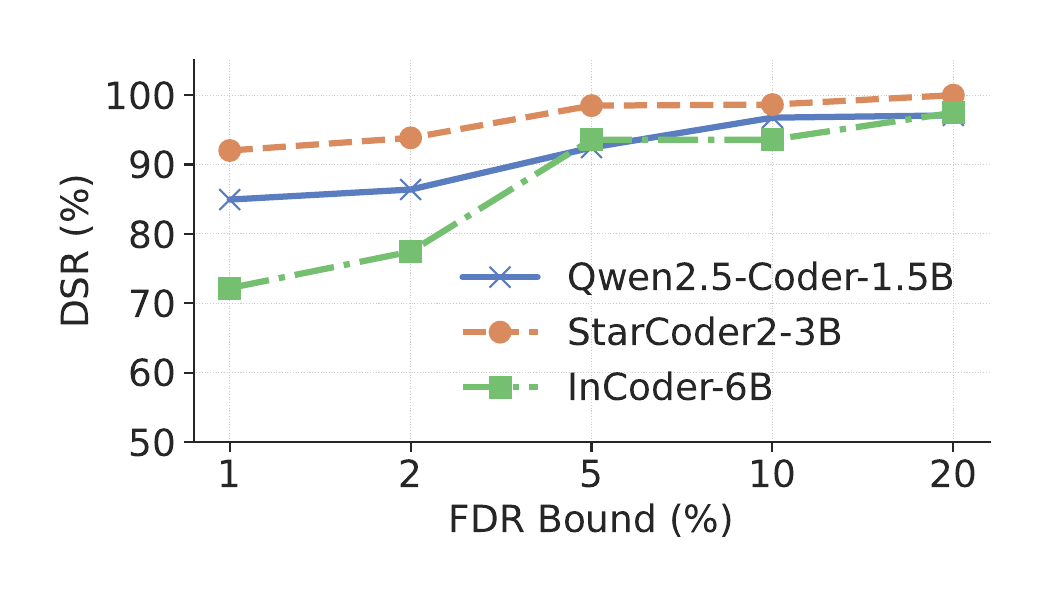}} \hspace{+1mm}
\subfloat[CodeSearchNet Dataset]{\includegraphics[width=0.32\textwidth]{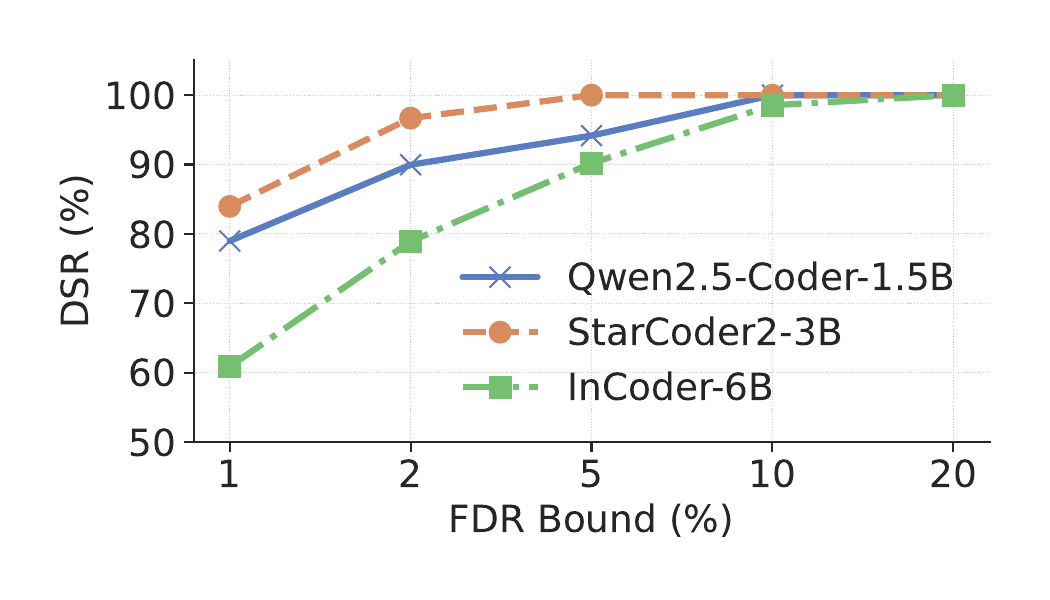}} 
\subfloat[CodeNet Dataset]{\includegraphics[width=0.32\textwidth]{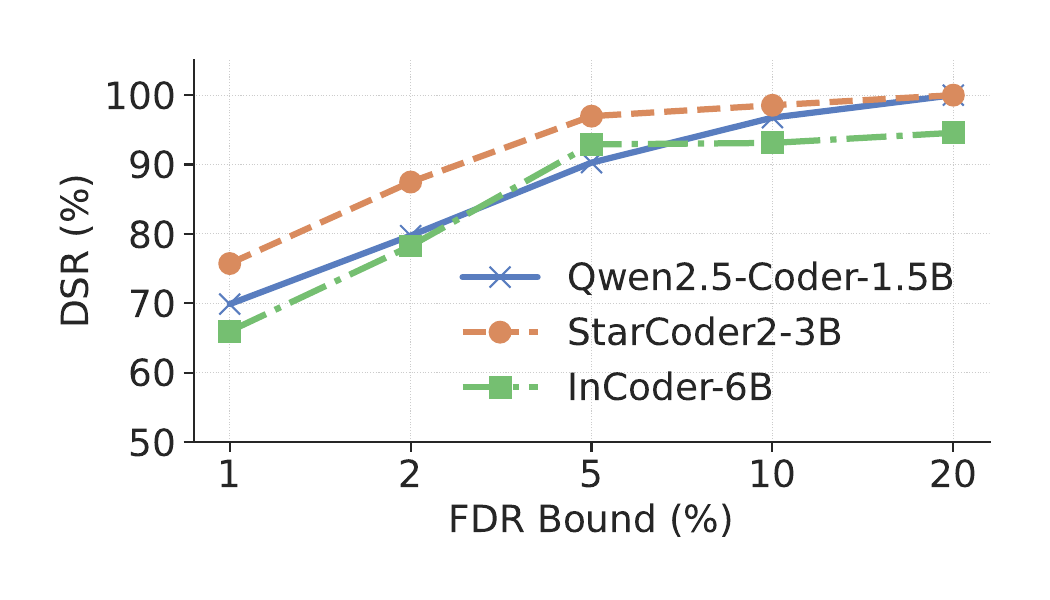}}   
 \caption{DSRs of \name across 3 models and 3 datasets under FDR guarantee levels of 1\%, 2\%, 5\%, 10\%, 20\%.}
 \label{fig:mainexp}
 \end{figure*}
\begin{table*}[t]
\centering
\caption{Comparing the DSR for different data auditing methods under 1\%, 2\%, 5\%, 10\% and 20\% FDR guarantee. The experiment is conducted on \revise{Qwen2.5-Coder-1.5B} and CodeParrot datasets.} 
\renewcommand{\arraystretch}{1.1}
\scalebox{1.05}{
\begin{tabular}{|l|c|c|ccccc|}
\hline
\multirow{2}{*}{} & \multirow{2}{*}{Category} & \multirow{2}{*}{\begin{tabular}[c]{@{}c@{}}FDR  Guarantee \end{tabular}} & \multicolumn{5}{c|}{DSR (\%)} \\ \cline{4-8}
           &  &  & \multicolumn{1}{c|}{1\%} & \multicolumn{1}{c|}{2\%} & \multicolumn{1}{c|}{5\%} & \multicolumn{1}{c|}{10\%} & 20\% \\ \hline
Loss attack~\cite{shokri2017membership} & \begin{tabular}[c]{@{}c@{}}Membership inference\end{tabular} & \ding{55} & \multicolumn{1}{c|}{0.97} & \multicolumn{1}{c|}{1.14} & \multicolumn{1}{c|}{2.53} & \multicolumn{1}{c|}{7.98} & 19.14 \\ \hline
min-k~\cite{mink} (ICLR'24) & \begin{tabular}[c]{@{}c@{}}Membership inference\end{tabular} & \ding{55} & \multicolumn{1}{c|}{3.50} & \multicolumn{1}{c|}{5.25} & \multicolumn{1}{c|}{8.37} & \multicolumn{1}{c|}{15.95} & 34.63 \\ \hline
zlib~\cite{carlini2021extracting} & \begin{tabular}[c]{@{}c@{}}Membership inference\end{tabular} & \ding{55} & \multicolumn{1}{c|}{1.53} & \multicolumn{1}{c|}{3.33} & \multicolumn{1}{c|}{8.56} & \multicolumn{1}{c|}{15.68} & 30.58 \\ \hline
Dataset inference~\cite{mainillm} (NeurIPS'24) & \begin{tabular}[c]{@{}c@{}}Membership inference\end{tabular} & \ding{55} & \multicolumn{1}{c|}{8.56} & \multicolumn{1}{c|}{15.12} & \multicolumn{1}{c|}{18.68} & \multicolumn{1}{c|}{29.79} & 54.67 \\ \hline
CodeMark~\cite{sun2023codemark} (FSE'23) & Backdoor marking & \ding{55} & \multicolumn{1}{c|}{0.32} & \multicolumn{1}{c|}{0.86} & \multicolumn{1}{c|}{1.28} & \multicolumn{1}{c|}{3.45} & 9.15 \\ \hline
\begin{tabular}[c]{@{}l@{}}Huang et al.~\cite{generalccs24} (CCS'24)\\ (with \name's renaming algorithm)\end{tabular} & Contrastive marking & \ding{51} & \multicolumn{1}{c|}{39.39} & \multicolumn{1}{c|}{43.31} & \multicolumn{1}{c|}{53.79} & \multicolumn{1}{c|}{63.63} & 77.22 \\ \hline
\name & Contrastive marking & \ding{51} & \multicolumn{1}{c|}{84.94} & \multicolumn{1}{c|}{86.39} & \multicolumn{1}{c|}{92.44} & \multicolumn{1}{c|}{96.75} & 97.07 \\ \hline
\end{tabular}
}
\label{tab:main}
\end{table*}


\noindent
\textbf{Models.} We conduct experiments on data marking using \revise{
Qwen2.5-Coder-1.5B~\cite{hui2024qwen2}}, StarCoder2-3B~\cite{lozhkov2024starcoder}, and InCoder-6B~\cite{incoder} as target models.   
Among them, \revise{Qwen2.5-Coder-1.5B is part of the latest Qwen2.5-Coder series, a family of large language models specialized for code understanding and generation. Models in this family achieve state-of-the-art performance among open-source code LLMs}. StarCoder2-3B, released in 2024, is a popular open-source code LLM that employs advanced attention mechanisms and a large context window. InCoder-6B, developed by Facebook, stands out for its ability to perform both standard left-to-right code generation and code infilling, making it a versatile tool for various coding tasks. For the oracle LLM leveraged by \name, we use another code LLM Yi-Coder-1.5B~\cite{01ai}.

\noindent\textbf{Datasets.} \revise{ Following prior work~\cite{lee2023wrote,code_mem1,code_mem2}, we primarily carry out our experiments on three commonly used code datasets: CodeParrot~\cite{codeparrot}, CodeSearchNet~\cite{husain2019codesearchnet}, and CodeNet~\cite{puri2021codenet}.} We focus on Python code, though our techniques are also applicable to other languages such as C++, Java, and Rust.  \revise{The CodeParrot, CodeSearchNet, and CodeNet datasets contain 410,210, 71,246, and 93,570 repositories, respectively. The file number of each repository ranges from 1 to 7,343, with an average of 12.28 files per repository in three datasets. }


\revise{In our experiments, we consider the \textit{training-from-scratch} setting. We randomly select 1\% of the repositories in the dataset as the target repositories to be protected. We inject marks into these repositories, while leaving the remaining repositories in the dataset unchanged. After training, we compute the DSR as the proportion of marked repositories that are successfully detected. 
Unlike dataset-level auditing methods~\cite{generalccs24,sun2023codemark} that output a single binary decision for the entire corpus, \emph{our marking and detection operate independently for repository}, yielding an individual decision for each. Consequently, our detection performance does not depend on the proportion of repositories selected for protection.  }

\noindent
\textbf{Metrics.} We evaluate \name along two dimensions: its detection accuracy, as well as its impact on code quality.  

For detection accuracy, we measure it using the DSR metric, which captures the fraction of marked repositories in the training set that are successfully identified. Following Huang et al.~\cite{generalccs24}, we measure the DSR under different FDR.  For data mark methods that have a theoretical FDR guarantee, we measure the DSR at different specified FDR guarantees; for methods that cannot provide a guarantee on FDR (e.g. membership inference methods), we measure the DSR at different specified empirical FDRs. It is worth noting that although under experimental conditions, we can compute empirical FDR for no-FDR-guarantee methods using member and non-member labels of all samples, such labels are not available in real-world auditing. Therefore, the empirical FDR values for these methods are not computable outside of the experimental setting.
 
 For code quality, we mainly measure the extent to which the data marking method alters the original code. We adopt three metrics: CodeBLEU, edit distance, and perplexity.  
 CodeBLEU~\cite{ren2020codebleu} is a classical metric to measure the quality of code adopted in previous works~\cite{wu2024codewmbench,yang2024srcmarker,li2024extracting}. It extends the traditional BLEU metric by incorporating syntax and semantics specific to programming languages, allowing for a more nuanced assessment of the code generated. It measures the similarity between the original code and the marked code, using a weighted combination of n-gram match, weighted n-gram match, AST match, and data-flow match scores. Edit distance~\cite{editdistance1,editdistance2} is a common metric used to measure the difference between two strings. It quantifies the minimum number of insertion, deletion, and substitution operations required to transform one string into another.  In our paper, to quantify the quality of marked code, we measure the edit distance between marked code and original code on a per-token basis. We also adopt the perplexity (PPL) metric, which is widely adopted in the NLP community, following previous works~\cite{ppl1,ppl2,qu2024provably,ppl3}. In our experiments, the PPL scores are computed using the Qwen2.5-Coder-32B~\cite{hui2024qwen2} model.

\noindent
\textbf{Hyperparameters.} In our experiments, we set the theoretical FDR guarantee to $p=5\%$ and the number of mark versions to $m=100$. 
The mark sparsity parameter is set to $K=100$ (i.e., on average, there is at most one mark per $100$ lines of code), and the mark position rank threshold is set to $R=500$.  For model training, the learning rate is initialized to $1\times10^{-5}$. \revise{To avoid overfitting, following ~\cite{touvron2023llama}, we set the number of training epochs to 2; in other words, the model sees each training code snippet only twice.}

\noindent
\textbf{AST Parser.}
We use the tree-sitter library~\cite{tree_sit_git,tree_sit_wiki} to parse the code and perform transformations on AST. 
    We first use tree-sitter to identify the \textit{function parameter names} and \emph{variable names} to inject our marks. After locating the marking positions, we utilize tree-sitter  to perform variable renamings on the AST. 

   

\subsection{Main results}

\noindent \textbf{\name performs well on different models and datasets.} We carry out the experiments with 3 different code models, \revise{Qwen2.5-Coder-1.5B}, StarCoder2-3B, and InCoder-6B, and 3 different datasets, CodeParrot, CodeSearchNet, and CodeNet. We measure the effectiveness of the data-usage detection of \name method under different FDR guarantees 1\%, 2\%, 5\%, 10\%, 20\%, and the results are shown in Figure~\ref{fig:mainexp}. We use default hyperparameter settings for these experiments. On the CodeParrot dataset, with FDR guarantee 5\%, \name 
 achieves \revise{92.44\%}, 98.64\%, and 93.58\% DSR on \revise{Qwen2.5-Coder-1.5B}, StarCoder2-3B, and InCoder-6B, respectively. These results show that we can detect more than 90\% code training usage if these models are trained on our marked code.  The performance of \name is robust across three different datasets. Even under the worst case, it still achieves a DSR higher than $85\%$ under FDR guarantee 5\%.  These results demonstrate the strong generalizability of \name under different models and datasets, indicating the potential value of deploying \name as a tool in real-world data auditing.

\noindent \textbf{\name outperforms existing baselines.} We also compare \name with existing data-usage auditing methods, which can be categorized into two types: (1) membership inference and (2) data marking. For membership inference-based methods, we compare against   several prevalent approaches for LLMs, including loss-based membership inference~\cite{shokri2017membership}, min-k inference~\cite{mink}, zlib-based membership inference~\cite{carlini2021extracting}, and dataset inference~\cite{mainillm}. For data marking-based methods, we compare against CodeMark~\cite{sun2023codemark} (FSE'23) and Huang et al.~\cite{generalccs24}. \revise{Since Huang et al.~\cite{generalccs24} do not provide a marking algorithm for the code domain, we adapt their framework by integrating the marking algorithm of \name into their pipeline.}

Since previous results demonstrate that our method performs consistently across different datasets and models, we conduct all baseline comparisons on the Qwen2.5-Coder-1.5B model and the CodeParrot dataset due to resource constraints. The comparison results are shown in Table~\ref{tab:main}. 
It is evident that \name consistently outperforms all baselines under different configurations.  Under 5\% FDR guarantee, \name achieves the highest DSR of  \revise{92.44\%}, significantly surpassing all membership inference-based methods and data marking methods. Under the same FDR, the best data marking method except ours can only achieve a DSR of 53.79\%, while the best membership inference-based method only achieves a DSR of 18.68\%. These methods fail to identify a large proportion of the trained repositories, whereas \name only misidentifies 10\% of the repositories used for training.


\noindent \textbf{\name achieves good imperceptibility.}
We measure \name's influence on code quality using CodeBLEU, edit distance, and PPL on CodeParrot dataset. 
According to Table~\ref{tab:quality_influence}, \name effectively preserves code quality with minimal changes, offering good marking imperceptibility. It achieves a high CodeBLEU score of 0.963. The change to the original code by \name is minimal, measured by edit distance. On average, only 3.61 tokens are modified per 100 lines of code. The PPL of our marked code is also very close to the PPL of unmarked code. This demonstrates that \name's modifications are unlikely to be noticed by human inspectors and remain mostly imperceptible to the model trainer.

 \begin{figure*}[t]
 \centering
\subfloat[Impact of $m$ \label{fig:impact_m}]{\includegraphics[width=0.32\textwidth]{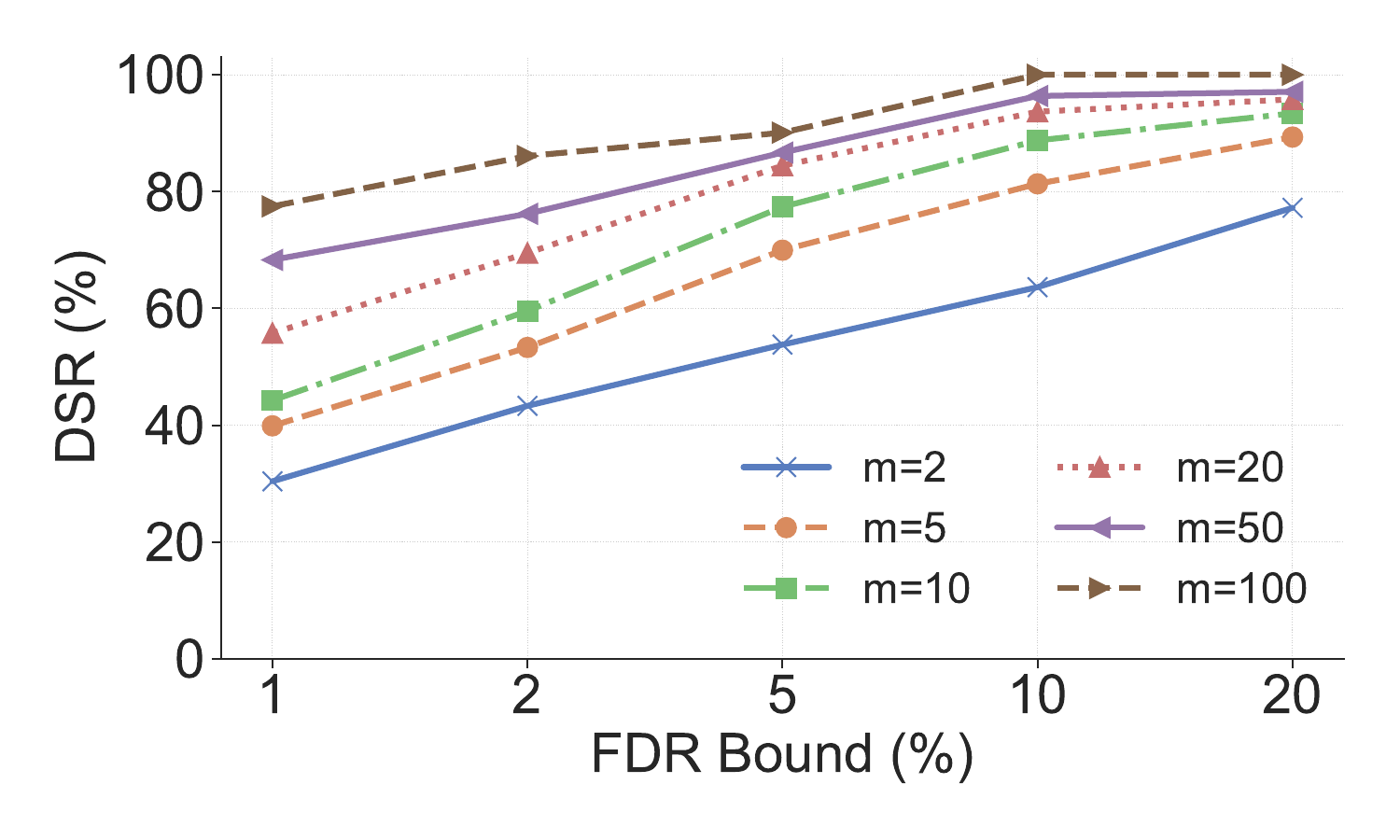}} \hspace{+1mm}
\subfloat[Impact of $K$ \label{fig:impact_K}]{\includegraphics[width=0.32\textwidth]{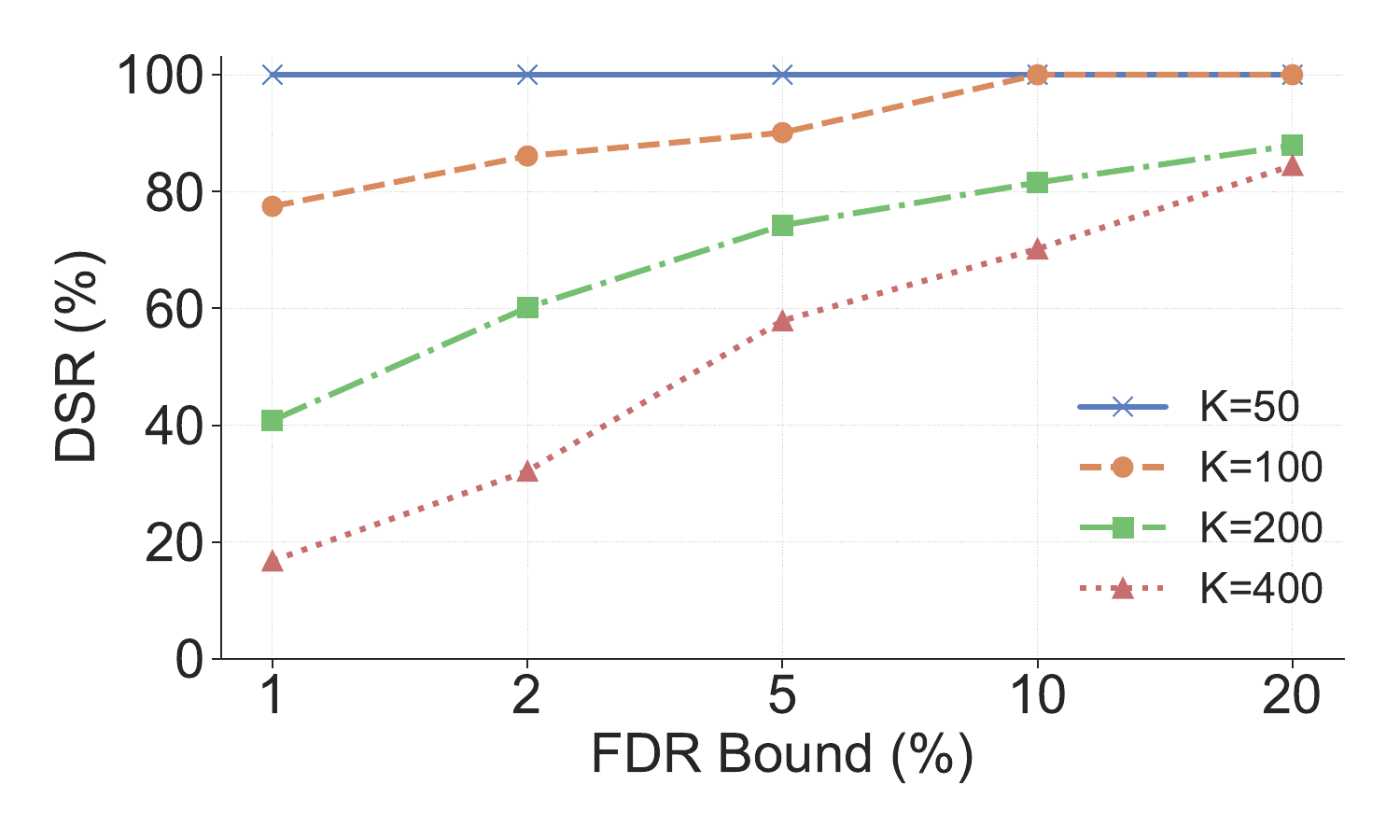}}   \hspace{+1mm}
\subfloat[Impact of $N$ \label{fig:impact_N}]{\includegraphics[width=0.32\textwidth]{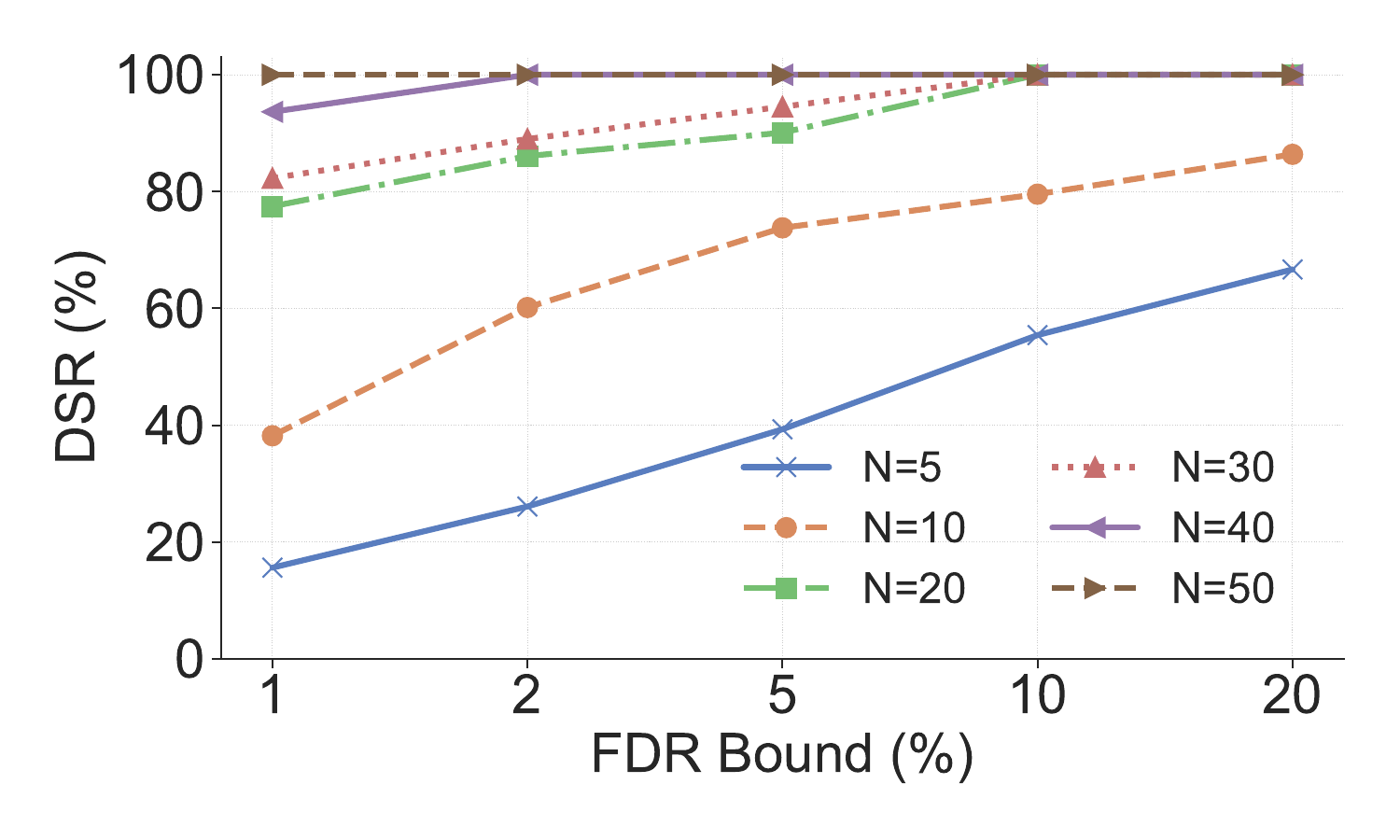}} 
 \caption{Impact of version number $m$, mark sparsity parameter $K$, repository file number $N$  on the DSR curve. 
 }
 \label{fig:ablation_m_size_K}
 \end{figure*}

\noindent \textbf{\name does not impact code LLM training.} In addition, we observe that the impact of \name on code LLM training is negligible. We employ the HumanEval dataset~\cite{chen2021evaluating} to evaluate the performance of \revise{Qwen2.5-Coder-1.5B} trained with a (1) normal CodeParrot dataset and (2) CodeParrot dataset injected by \name. Following HumanEval ~\cite{chen2021evaluating}, we measure the model performance with metric Pass\texttt{@} $k$~\cite{chen2021evaluating}, which measures the proportion of samples for which at least one of $k$ independently generated outputs passes the predefined unit tests. According to Table~\ref{tab:clean-performance}, training on dataset marked by \name only has negligible influence on model performance measured by Pass\texttt{@}$1$, Pass\texttt{@}$10$ and Pass\texttt{@}$100$, respectively. Therefore, the model trainer is very unlikely to observe the existence of our mark via the code model's capability.

 \begin{table}[t]
\centering
\caption{The influence of \name algorithm on code quality, evaluated on the CodeParrot Dataset. }
\scalebox{1.1}{
\begin{tabular}{|c|c|c|c|}
\hline
\multicolumn{1}{|l|}{} & Edit distance & PPL  & \multicolumn{1}{c|}{CodeBLEU} \\ \hline
Unmarked               & 0             & 1.04 & 1                             \\ \hline
\name             & 3.61          & 1.11 & 0.9633                        \\ \hline
\end{tabular}
}

\label{tab:quality_influence}

\end{table}

\subsection{Ablation study}
\label{sec:ablation}

In this section, we discuss the key factors that influence the detection performance of \name, including training parameters, dataset structure, and hyperparameters of our method.

\noindent\textbf{Impact of $m$.} We first study the impact of the number of marking versions $m$ on DSR. 
As illustrated in Figure~\ref{fig:impact_m}, DSR of our detection method improves as $m$ increases. Under the 5\% FDR guarantee, the DSR increases steadily from 53.79\% under $m=2$ to \revise{92.44\%} under $m=100$. This is not surprising because when the number of versions $m$ increases, each position's rank delivers more nuanced information, enabling more accurate detection. Another observation is that the effect of increasing $m$ on the detection accuracy is gradually diminishing: under the 5\% FDR guarantee, the accuracy only increases by 5.64\% (from 84.44\% to \revise{92.44\%}) when $m$ increases from 20 to 100. This indicates that when $m$ is sufficiently large, the information gain from generating
 more code variants (i.e., larger $m$) is less significant.

 \begin{table}[t]
\centering
\caption{The influence of $K$ on code quality of \name on the CodeParrot Dataset.}
\scalebox{1.1}{
\begin{tabular}{|lc|c|c|c|}
\hline
\multicolumn{1}{|c|}{} & $K$                                                          & Edit distance & PPL  & CodeBLEU \\ \hline
\multicolumn{1}{|l|}{\multirow{4}{*}{Impact of  $K$}}        & $50$   & 6.74          & 1.20 & 0.8059   \\ \cline{2-5} 
\multicolumn{1}{|l|}{}                                     & $100$  & 3.61          & 1.11 & 0.9633   \\ \cline{2-5} 
\multicolumn{1}{|l|}{}                                     & $200$  & 1.89          & 1.10 & 0.9842   \\ \cline{2-5} 
\multicolumn{1}{|l|}{}                                     & $400$  & 0.97          & 1.09 & 0.9992   \\ \hline
\end{tabular}
}

\label{tab:quality_K}
\end{table}

 \begin{table}[t]
\centering
\caption{The model performance of Qwen2.5-Coder-1.5B trained on the CodeParrot dataset with or without \name, measured on the HumanEval dataset~\cite{humaneval}.  }
\scalebox{1.1}
{
\begin{tabular}{|l|c|c|}
\hline
Metric  &  Without mark & With mark \\ \hline
Pass@1 & 3.44\% & 3.38\%
 \\ \hline
Pass@10  & 7.80\% & 7.72\%
 \\ \hline
Pass@100   & 15.24\% & 15.11\%
 \\ \hline
\end{tabular}
}
\label{tab:clean-performance}
\end{table}

\noindent\textbf{Impact of $K$.} We further evaluate how the mark sparsity $K$ would influence the DSR of \name,  and the experimental results can be found in Figure~\ref{fig:impact_K}. It can be observed that as $K$ increases, the DSR exhibits a gradual decline while \name still remains usable. When $K=50$, our detection method achieves 100\% detection accuracy under all FDR guarantees. With $K$ increasing to 200, the DSR under 5\% FDR guarantee is still higher than 70\%, indicating that our detection method demonstrates high robustness regardless of $K$. It is worth noting that while using a smaller $K$ can improve detection performance, it also leads to deterioration in code quality. As shown in Table \ref{tab:quality_K}, $K$ affects edit distance in a nearly proportional way. When $K \geq 100$, both PPL and CodeBLEU remain quite close to the unmarked case, indicating that code quality is not significantly impacted. In contrast, when $K$ is reduced to 50, the  CodeBLEU drops to 0.8059, which indicates a noticeable degradation in code quality. In conclusion, $K=100$ is a proper choice to balance the trade-off between code quality and detection performance.

\noindent\textbf{Impact of repository size.} We also evaluate the impact of repository file number $N$ on the detection performance of \name (under this case we only consider auditing repositories with file number $N$). The results are shown in Figure~\ref{fig:impact_N}.   The detection accuracy increases with the growth of $N$. When the repository file number increases to 40, DSR increases to over 90\% even under 1\% FDR guarantee. When the repository file number increases to $N\ge 50$, the DSR is 100\%, which indicates that \name successfully identify all repositories used in training under all tested FDR guarantees.     This trend aligns with our expectations because a larger repository file number enables more marking positions for detection, thus leading to a higher detection accuracy. These results highlight the effectiveness of \name in auditing commercial-scale code repositories.      At the same time, \name also demonstrates strong capabilities in handling repositories with small sizes. Notably, \name can achieve 73.79\% DSR for an extremely small repository size ($N=10$) under a 5\% FDR guarantee,  demonstrating its strong capability in protecting the copyright of ordinary users’ code.



\subsection{Deployment overhead}
In this section, we discuss and explore the deployment overhead of our algorithm.

\noindent \textbf{Marking overhead.} Our marking algorithm is highly efficient because the main computational overhead of marking each file lies in a single forward pass of each file through the oracle code LLM, during which full vocabulary logits are computed for all positions.  We measured that marking one repository with 20 files takes 15.6 seconds on average. 

\noindent \textbf{Storage overhead.} The storage overhead of our algorithm is also minimal. For detection purposes, we only need to store the original token (prior to marking) and its position within the file. This information is sufficient to reconstruct the list of alternative tokens at each marked position. Under our default setting of $K = 100$—i.e., at most one mark per 100 lines of code—the average number of marked positions per file is about 3. Each marked position requires 4 bytes of storage: 2 bytes for the token ID and 2 bytes for its position. Therefore, for a large repository with 100 files, the total storage overhead is approximately 1.2KB, which is negligible.

\noindent \textbf{Detection overhead.} Our detection process primarily involves computing the relative logit ranks of the $m$ token variants at each marking position. The logits vector at each marking position can be obtained through a single forward pass of the code LLM. Consequently, the overall detection cost for a repository scales linearly with both the total number of tokens in the target repository and the number of marking positions in each code file. 

For instance, a large commercial repository such as the React library~\cite{react} contains approximately 4 million tokens. Assuming each code file has, on average, three marking positions, the total cost corresponds to forwarding around 12 million tokens. Given that state-of-the-art LLM APIs (e.g., GPT-5) charge no more than $10^{-5}$ USD per token~\cite{oai_api}, the total detection cost for a repository of this scale is estimated to be under 120 USD—demonstrating that \name remains both computationally and economically practical even for large projects.


\noindent \textbf{Version control strategy.} Since code files in real-world repositories are frequently updated, a version control strategy can be adopted to avoid reinjecting marks after every change. As long as the marked tokens remain unmodified, we simply update the stored marking positions to reflect the changes. New mark injection is only required when previously marked tokens are deleted, or when the user adds a substantial amount of new code that creates additional capacity for marking. In real-world auditing, to reduce detection overhead, we can choose to run the detection algorithm only on each “major version” of the repository—i.e., versions that differ substantially from the previous ones. Our detection algorithm’s FDR guarantee still holds for each individual detection, although the results across different detections may be statistically dependent.

\subsection{Potential countermeasures of the model trainer}

In this section, we explore several potential countermeasures against our marking algorithm.



    
Dataset filtering~\cite{chen2018detecting, tran2018spectral, doan2020februus, huang2022backdoor} is a class of techniques that were initially proposed to identify and remove backdoored data from the training corpus. Following CodeMark~\cite{sun2023codemark}, we evaluate two typical dataset filtering strategies, activation clustering~\cite{chen2018detecting} and spectral signature~\cite{tran2018spectral}.  
The model trainer will use the aforementioned methods to remove the code files that are identified as marked from the training dataset.  
To evaluate the effectiveness of the dataset filtering method, we compute both the proportion of unmarked code removed from the unmarked dataset and the proportion of marked code removed from the marked dataset.    
 The results are shown in Table~\ref{tab:countermeasure}. It can be observed that both activation clustering and spectral signature are ineffective in removing our mark. For both methods, their removal ratio on marked code is similar to unmarked code, indicating that they essentially performs random guessing on which code files to remove code within the dataset. Moreover,  \name still achieves 78.4\% and 86.6\% DSR under the two removal strategies, respectively.

\begin{table}[t]
\centering
\caption{The effect of two dataset filtering strategies.}
\scalebox{0.88}
{
\begin{tabular}{|l|l|cc|c|}
\hline
\multirow{2}{*}{Removal strategy} & \multirow{2}{*}{Data marking} & \multicolumn{2}{c|}{Remove ratio (\%)}      & \multicolumn{1}{l|}{\multirow{2}{*}{DSR (\%)}} \\ \cline{3-4}
                                  &                               & \multicolumn{1}{c|}{Unmarked} & Marked & \multicolumn{1}{l|}{}                          \\ \hline
\multirow{2}{*}{Activation Clustering~\cite{chen2018detecting}}               & CodeMark~\cite{sun2023codemark}                      & \multicolumn{1}{c|}{0.45}     & 0.43   & 0.0                                         \\ \cline{2-5} 
 & \name                         & \multicolumn{1}{c|}{0.45}     & 0.46   & 78.4         \\ \hline
\multirow{2}{*}{Spectral Signature~\cite{tran2018spectral}}             & CodeMark~\cite{sun2023codemark}                     & \multicolumn{1}{c|}{0.15}     & 0.14   & 0.8          \\ \cline{2-5}         & \name                          & \multicolumn{1}{c|}{0.15}     & 0.13   & 86.6                                        \\ \hline
\end{tabular}
}
\label{tab:countermeasure}
\end{table}

\begin{table}[t]
\centering
\caption{The effect of the model trainer's variable renaming strategy with different renaming ratios. We test under three choices of the attacker's oracle code LLM. }
\scalebox{1.1}
{
\begin{tabular}{|l|c|c|c|c|}
\hline
Renaming ratio   & 25\%                       & 50\%                       & 75\%                       & 100\%                       \\ \hline
\revise{Deepseek-Coder-1.3B~\cite{guo2024deepseek}} & 92.07                     & 81.15                     & 77.32                    & 66.34               \\ \hline
StableCoder-3B~\cite{pinnaparaju2024stable}    & 90.68                     & 88.22                    & 78.79                      & 70.52                   \\ \hline
StarCoder-3B~\cite{starcoder}     & \multicolumn{1}{l|}{89.27} & \multicolumn{1}{l|}{85.28} & \multicolumn{1}{l|}{75.60} & \multicolumn{1}{l|}{66.73} \\ \hline
\end{tabular}
}
\label{tab:trainer-remove}
\end{table}

We also evaluate variable renaming as an adaptive attack strategy against \name, where the model trainer (attacker) randomly renames a certain proportion of variables to neutralize the embedded marks. The core idea is to aggressively rename variables in the training data to ``cover'' those renamed by \name. To enhance the attack, the adaptive trainer can leverage an oracle code LLM to select variables for renaming that fall within the same logits rank range targeted by our marking algorithm. Importantly, we assume the attacker does not know which oracle model each user employs, as the training dataset consists of code written by a large number of users, each potentially using a different oracle code LLM.

We evaluate this strategy under renaming ratios of 25\%, 50\%, 75\%, and 100\%, with detection results presented in Table~\ref{tab:trainer-remove}. When the renaming ratio is at or below 75\%, the impact on \name's detection performance remains minimal—across different oracle code LLMs, our system continues to achieve DSR $> 75\%$ under 5\% FDR guarantee. Even when the attacker renames all tokens proposed by the oracle code LLM (i.e., 100\% renaming), \name still successfully identifies over 65\% of the training repositories. This robustness stems from the variation in logits rankings across different oracle models, making it difficult for an attacker using oracle A to accurately obscure the marks generated by oracle B. 

Notably, such aggressive renaming alters code context and harms code quality and readability. To evaluate the utility loss, we trained the \revise{Qwen2.5-Coder-1.5B} model with a 100\% renaming ratio. The model’s Pass@100 score dropped from 15.24\% (no renaming) to 12.78\%, indicating that excessive variable renaming impairs model training. This poses a significant challenge for attackers: without knowledge of the exact oracle used, it is difficult to suppress \name effectively and maintain model performance simultaneously.

%% file: discuss.tex
\section{Discussion and Limitation} 

\noindent \textbf{Auditing existing code LLMs.} 
To the best of our knowledge, membership inference is the only existing data-auditing approach that can be applied to already-trained models. However,  membership inference achieves substantially lower detection accuracy than data marking methods such as \name, due to a lack of controlled injection before training. Enhancing the performance of membership inference remains an important but orthogonal research direction to ours.

%% file: conclusion.tex
\section{Conclusion}
The widespread deployment of code LLMs has raised pressing ethical and legal concerns regarding the unauthorized use of open-source repositories. In this paper, we present \name, a proactive and theoretically grounded data-auditing framework tailored to the code domain. Our framework simultaneously achieves semantic preservation, imperceptibility, data efficiency, and a provable FDR guarantee—four properties that have not been satisfied together by any prior work. By generating multiple semantically equivalent code variants and employing a rank-based hypothesis test over model responses, \name can reliably detect whether a deployed code LLM has been trained on a given repository. Comprehensive experiments across diverse datasets and model architectures demonstrate that \name consistently delivers high detection accuracy, substantially outperforming state-of-the-art data-usage auditing baselines, even when the target repository contains only a small number of code files. Moreover, its imperceptible variable-renaming strategy ensures practical robustness against model trainers attempting to remove the marks.  By enabling code authors to reliably verify the usage of their repositories, \name contributes to building a more transparent and accountable ecosystem for code LLM  development.

%% file: appendix.tex
\appendix


\clearpage

%


\begin{figure*}[h]

\centering
\subfloat{\includegraphics[width=1\textwidth]{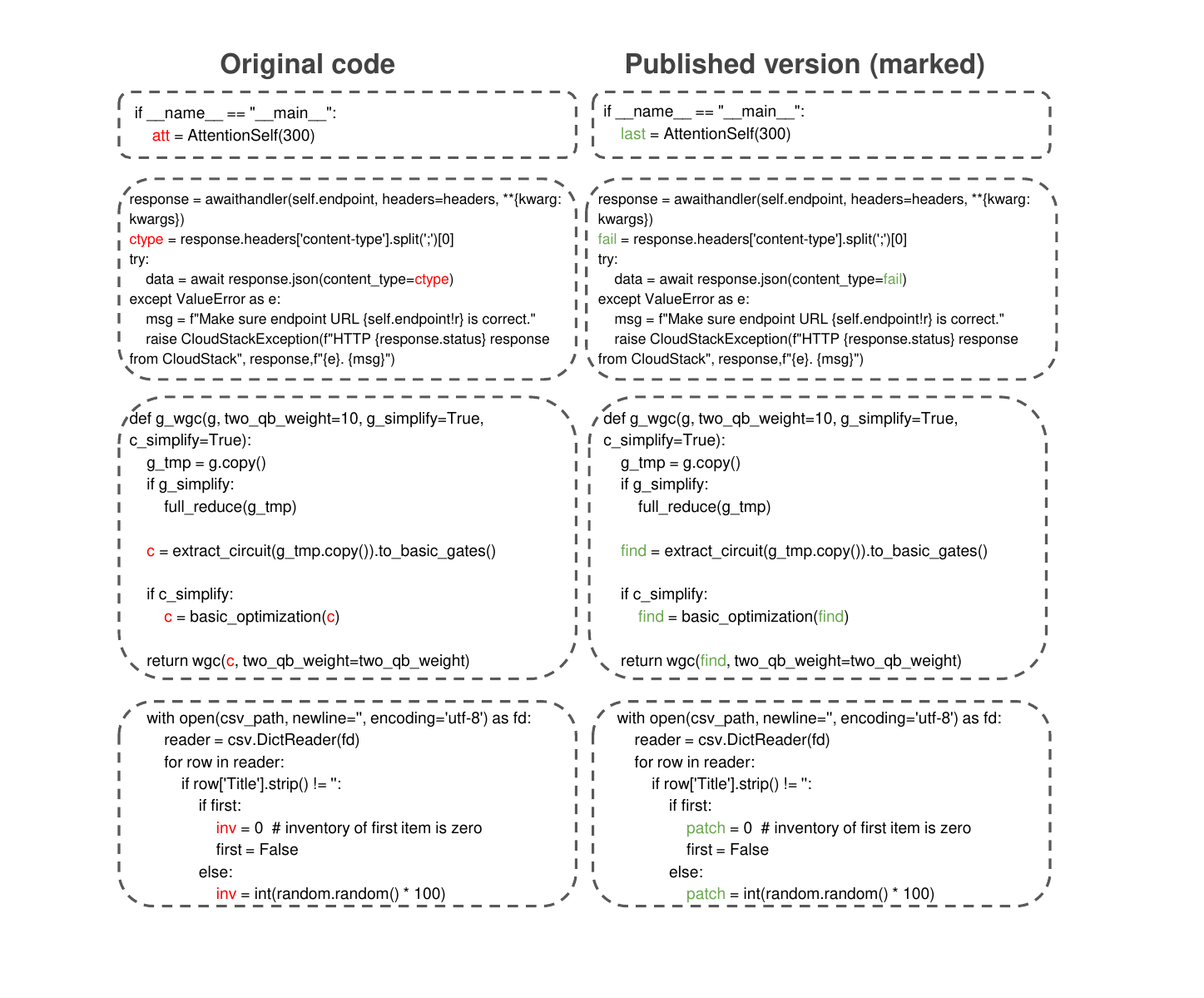}}
\caption{More examples of code marked with our method. }
\label{fig:moreexample}
\end{figure*}

\begin{algorithm}[h]
   \caption{\emph{Marking algorithm for single code file}}
   \label{alg:replacement}
\begin{algorithmic}
   \STATE {\bfseries Input:} A code file $x$; line number of code $L$; mark sparsity $K$; oracle language model $\mathsf{O-LLM}$; ranking threshold $R$;  number of versions $m$
   \\
     \STATE {\bfseries Output:} A marked code file $\hat{x}$, injected mark positions $\delta$, original tokens $V$\\
    \STATE Extract all the local variables by AST parser,  record their positions as set $\vartheta$;\\
    \STATE $\delta=\{\}$;\\
    \FOR{$i=1,2,\cdots,|\vartheta|$}
    \STATE  $v=\mathsf{O-LLM}(x_1,\cdots,x_{\vartheta_i-1})$;\\ 
    \IF {$\mathsf{rk}(v,v[x_{\vartheta_i}]) \ge R$}
    \STATE $\delta.\mathsf{insert}(\vartheta_i)$ \\
    \ENDIF
    \ENDFOR
    \IF {$|\delta| \ge \frac{L}{K}$}
    \STATE $\delta=\{\text{Randomly select } \lfloor\frac{L}{K}\rfloor \text{different elements from } \delta\}$;
    \ENDIF
  \STATE   $\hat{x}=x$\\
  \STATE $V=\{\}$\\
    \FOR{$i=1,2,\cdots,|\delta|$}
    \STATE  $v=\mathsf{O-LLM}(x_1,\cdots,x_{\delta_i-1})$;\\ 
    \STATE $V.\mathsf{insert}(x_{\delta_i})$\\
    \STATE Construct the similar variable name set:\\
    \STATE $P=\mathsf{rk}(v,v[x_{\delta_i}])$; \\
    $B=\{t|\mathsf{rk}(v,v[t])\in [P-m/2, P+m/2-1]\}$  \\
    \STATE Randomly select a variable name $t$ from $B$ \\
    \STATE Rename variable $\hat{x}[\delta_i]$ by name $t$ via the AST parser at all occurrences of the corresponding variable in the code file $\hat{x}$. \\
    \ENDFOR
    \RETURN $\hat{x},\delta,V$
\end{algorithmic}
\end{algorithm}

\begin{algorithm}[h]
   \caption{\emph{Detection algorithm}}
   \label{alg:detection}
\begin{algorithmic}
   \STATE {\bfseries Input:} Target code LLM for auditing $\mathsf{A-LLM}$; A set of marked code files $F$ in a repository; Repository size $N$; A nested set $\delta$ of injected mark numbers of different files; A nested set $V$ of original names of different files;  Oracle code LLM $\mathsf{O-LLM}$; number of versions $m$; desired false-detection rate bound $p$.
   \\
     \STATE {\bfseries Output:} A binary value $b$, $b=1$ means detected data-usage in the repository, $b=0$ means not detected\\
     \STATE Compute number of injected marks: $n=\sum_{i=1}^{|F|} |\delta_i|$\\
     \STATE Define the generating function of rank sum $S$ under $H_0$:
     \begin{align*}
     G_n(x) &= \left(  \frac{x + x^2 + \dots + x^m}{m}  \right)^n 
     \end{align*}
     \STATE Compute PDF of $S$:
\begin{align*}
\Pr(S = t) &= [x^t] G_n(x),~ t=m, \cdots, n\cdot m
\end{align*}
    \STATE Use binary-search to find rank-sum threshold $T$, s.t.:
    \begin{align*}
T= \arg\max_b \left(\sum_{t=1}^b \Pr(S=t) \right)\le p
\end{align*}
\STATE Rank sum $S=0$\\
\FOR {$i=1,2,\cdots, N$}
\FOR {$j=1,2,\cdots, |\delta_i|$}
 \STATE  $v=\mathsf{O-LLM}(F_{i,1}, \cdots,F_{i,\delta_{i,j}-1})$;\\
 \STATE  $h=\mathsf{A-LLM}(F_{i,1}, \cdots,F_{i,\delta_{i,j}-1})$; \\ 
     \STATE Reconstruct the similar variable name set:\\
      $P=\mathsf{rk}(v,v[V_{i,j}])$; \\
     $B=\{t|\mathsf{rk}(v,v[t])\in [P-m/2, P+m/2-1]\}$ \\
     \STATE Compute rank for detection:\\
     \STATE $S+=\mathsf{rk}(\{h_i | i\in B\},h[F_{i,\delta_{i,j}}])$;\\
\ENDFOR
\ENDFOR
\IF {$S\le T$}
    \RETURN $1$
\ELSE 
\RETURN $0$
\ENDIF


\end{algorithmic}
\end{algorithm}